\documentclass[a4paper, 12pt]{amsart}

\usepackage[utf8]{inputenc}
\usepackage{amsmath,amsfonts,amssymb,amsopn,amscd,amsthm}
\usepackage{comment}
\usepackage{dsfont}
\usepackage{graphicx}
\usepackage{color}
\usepackage[colorlinks]{hyperref}
\usepackage{epigraph}
\usepackage{todonotes}
\usepackage[left=3.5cm,top=2.5cm,bottom=2cm,right=3cm]{geometry}
\usepackage{tikz-cd}
\usepackage{tikz}
    \usetikzlibrary{arrows.meta,positioning,fit,calc}
\usepackage{caption}
\usepackage{listings}
\date{15th of October 2021}

\allowdisplaybreaks[4]

\DeclareRobustCommand{\SkipTocEntry}[5]{}

\DeclareMathOperator{\Var}{Var}
\DeclareMathOperator{\Cov}{Cov}

\def\1{{\mathbf 1}}

\newtheorem{thm}{Theorem}[section]
\newtheorem{proposition}[thm]{Proposition}
\newtheorem{corollary}[thm]{Corollary}

\newtheorem{definition}[thm]{Definition}

\newtheorem{lemma}[thm]{Lemma}

\newtheorem{remark}[thm]{Remark}

\numberwithin{equation}{section}
\numberwithin{figure}{section}
\date{\today}

\begin{document}

\title{A continuous-time Kyle model with price-responsive traders}

\author{Eunjung Noh}
\address{Department of Mathematics and Computer Science, Rollins College, Winter Park, Florida 32789, USA}
\email{enoh@rollins.edu}

\begin{abstract}
Classical Kyle-type models of informed trading typically treat noise trader demand as purely exogenous. In reality, many market participants react to price movements and news, generating feedback effects that can significantly alter market dynamics. This paper develops a continuous-time Kyle framework in which two types of price-responsive traders (momentum and contrarian traders) adjust their demand in response to price signals. This extension yields a finite-dimensional Kalman filter for price discovery and leads to a forward–backward Riccati system characterizing equilibrium. We show that when feedback is weak, equilibrium exists and is unique as a smooth perturbation of the classical Kyle solution, allowing us to derive explicit comparative statics for insider profits and price informativeness. For stronger feedback, the model generates rich dynamics, including potential multiplicity of equilibria and amplification effects. Our framework thus bridges the gap between purely exogenous noise and more realistic, behaviorally motivated trading.
\end{abstract}

\keywords
{Market microstructure, Kyle model, price-responsive traders, filtering}

\subjclass{91G80, 91B26, 93E11}
\maketitle

\section{Introduction}
\label{intro}

Noise traders play a central role in modern models of financial markets by masking informative order flow and thereby enabling trade in the presence of asymmetric information. In the classical Kyle model \cite{kyle} and its continuous-time extension by Back \cite{back1992}, noise trading is modeled as purely exogenous and price-inelastic: random orders arrive independently of fundamentals, prices, or past market activity. This assumption delivers analytical tractability and a sharp characterization of equilibrium, but abstracts from a key empirical observation: a substantial fraction of order flow in real markets reacts mechanically to price movements.

A large and growing empirical literature documents that many market participants, including trend-following hedge funds, certain classes of retail traders, and portfolio rules embedded in algorithmic execution, systematically adjust their holdings in response to recent price innovations rather than to underlying fundamentals. Such feedback trading generates endogenous correlation between prices and future order flow, influences the rate at which information is impounded into prices, and can amplify or dampen volatility depending on the strength and direction of the feedback. Despite its importance, the dynamic interaction between informed trading, price-responsive noise traders, and rational market-maker learning remains poorly understood in continuous time.

This paper develops a continuous-time Kyle model with endogenous, price-responsive noise traders. Instead of following a purely exogenous Brownian motion, aggregate noise trading incorporates mechanical responses to price changes through two economically distinct channels:

- \textit{momentum traders}, who buy following price increases and sell after declines, and

- \textit{contrarian traders}, who trade in the opposite direction.

These traders do not observe fundamentals and do not behave strategically. Their positions evolve according to linear mean-reverting dynamics driven by price innovations, allowing their behavior to be embedded within the inferential structure of the Kyle framework.

A central challenge in introducing endogenous feedback into continuous-time microstructure models is preserving finite-dimensional filtering. In the classical Back model, the market maker infers the latent fundamental directly from aggregate order flow, and the filtering problem remains one-dimensional. When order flow itself reacts to prices, the state space becomes endogenous and, a priori, infinite dimensional. This paper shows that by modeling momentum and contrarian positions as linear controlled diffusions, the market maker's inference problem remains linear-Gaussian and can be solved through a finite-dimensional Kalman-Bucy filter. The resulting conditional covariance matrix satisfies a matrix Riccati equation, which couples with the insider's optimal trading problem to produce a forward-backward Riccati system characterizing equilibrium.

This modeling innovation allows for a fully analytic treatment of feedback effects in continuous time. The presence of price-responsive traders distorts the informational content of order flow: prices now reflect not only private information revealed gradually through inconspicuous insider trading, but also mechanical trading that responds to past prices. The resulting feedback loop modifies the strength and timing of price discovery and introduces new stability considerations absent in the classical Kyle setting.

The main contributions of the paper are as follows.
First, we introduce a new continuous-time microstructure model in which aggregate noise trading depends on price innovations. Despite the endogeneity, the system remains linear–Gaussian, and all inference takes place in a finite-dimensional state space.
Second, when feedback parameters are small, equilibrium exists and is unique as a smooth perturbation of the classical Kyle solution. We derive explicit first-order comparative statics for price informativeness, insider trading intensity, and expected insider profits, showing how momentum or contrarian behavior affects information revelation through the filtering channel rather than through ad hoc changes in noise volatility.
As feedback strengthens, on the other hand, the stability properties that guarantee the Kyle equilibrium begin to fail. We identify three mathematically distinct breakdown mechanisms:
(i) finite-time divergence of the Riccati flow;
(ii) loss of contraction of the equilibrium fixed-point mapping; and
(iii) instability of the Kalman–Bucy filter.
These failures correspond to economically meaningful amplification effects and generate the possibility of equilibrium multiplicity.

Taken together, these results bridge the gap between classical Kyle models with exogenous noise and behavioral or algorithmic trading models with endogenous feedback. By preserving analytical tractability while incorporating mechanically price-responsive order flow, the model provides a flexible platform for studying learning, information aggregation, and stability in markets where nonstrategic traders react to prices.


The next section reviews prior studies. Section~\ref{sec:Model} details the model. In Section~\ref{sec:EqAnalysis}, by exploiting the linear-Gaussian structure, we use Kalman-Bucy filter to obtain a coupled forward-backward system of equations that characterizes equilibrium. In Section~\ref{sec:WeakFeedback}, for small feedback parameter, we prove local existence and uniqueness of equilibrium, and provide a sufficient local condition for stability. In Section~\ref{sec:StrongFeedback}, amplification and potential multiplicity of equilibria is possible when feedback is strong. Then, we conclude in Section~\ref{sec:Conclusion}.

\section{Literature review} \label{LitReview}
The starting point of our analysis is the seminal model of Kyle \cite{kyle}, in which a privately informed insider trades strategically against competitive market makers who observe only aggregate order flow from noise traders. In continuous time, Back \cite{back1992} shows that the insider's optimal strategy is inconspicuous, that prices reveal the fundamental gradually, and that equilibrium can be characterized through a scalar Riccati equation for the conditional variance of the asset value. 
Many other extensions of the problem while retaining the exogeneity of noise trading have been studied \cite{
back2000imperfect,cel,bal,bp,baruch,bose2020kyle,bose2021multidimensional,cc2007,
ccd,c,cd1,cd,
cen,
first_paper,ckl23,cdf,corcuera2010kyle,ekren2022kyle}.

A core limitation of this literature is that noise traders are modeled as price-inelastic and behaviorally featureless. The present paper extends the Kyle paradigm by introducing endogenous price-responsive noise traders while preserving continuous-time tractability through a finite-dimensional Kalman–Bucy filtering structure. This departs from existing extensions in which tractability relies heavily on exogenous Brownian noise.

A second strand of literature examines price–order flow feedback. Cespa and Vives \cite{CespaVives2012} show that prices serve as both information carriers and coordination devices, potentially generating amplification and multiplicity when investors place weight on public signals. Relatedly, He and Krishnamurthy \cite{He2013}, Xiong \cite{Xiong2001}, and Brunnermeier \cite{Brunnermeier2001} study feedback loops in which beliefs or risk constraints cause price movements to reinforce themselves. Collin-Dufresne and Fos \cite{cdf} analyzes equilibria with stochastic liquidity, generating time-varying market depth and rich price–volume comovement. However, these frameworks do not specify mechanical trading rules tied directly to price innovations, nor do they modify the filtering structure in which private information is inferred from order flow.

Behavioral finance challenges the classical view of noise traders as passive liquidity providers, emphasizing instead that such agents often behave as if they have information when they do not \cite{Black1986}. These traders commonly rely on technical rules or ``popular models” \cite{Shiller1984,Shiller1990} and may exhibit price-responsive feedback behavior, including positive-feedback or trend-chasing strategies \cite{DeLong1990b} as well as misguided contrarianism. Such trading induces serial correlation in returns \cite{Sentana1992} and, as agent-based simulations show, can amplify volatility and distort prices \cite{Dai2025}.
While these studies demonstrate the prevalence of price-based trading, they rely on discrete-time or reduced-form settings. The present paper incorporates these momentum and contrarian forces directly into a structural continuous-time microstructure equilibrium, showing how mechanical rules alter the filtering channel through which private information is reflected in prices.


\section{Model}\label{sec:Model}
We consider a continuous-time version of the Kyle (1985) market in which a single risky asset is traded over the time interval $[0, T]$. All stochastic processes are defined on a filtered probability space
$(\Omega,\mathcal F,(\mathcal F_t)_{t\ge0},\mathbb P)$ satisfying the usual conditions. The terminal fundamental value of the asset, denoted $v$, is a normally distributed random variable with mean $0$ and variance $\sigma_v^2$. Three types of agents interact in the market: an informed trader, a population of price-responsive noise traders, and a competitive market maker.

The informed trader observes the value $v$ at time $0$ and trades continuously at a rate $\theta_t$. Her cumulative position satisfies
	$$d \Theta_t = \theta_t dt, \quad \Theta_0 = 0 \ .$$

In contrast to the classical setting, we assume that there are two distinct types of price-responsive traders: \emph{trend followers (momentum traders)} and \emph{contrarians}. Let $Z_t$ denote their cumulative order flow, which evolves according to
$$dZ_t = \mu_t dt  + \sigma_z dW_t \ , $$
where $W_t$ is a standard Brownian motion capturing exogenous random orders, and the drift $\mu_t$ represents systematic (possibly price-dependent) behavior, with
\begin{equation}\label{eq:mu}
    \mu_t = \gamma_F m_t + \gamma_C c_t + \sigma_\varepsilon \varepsilon_t \ .
\end{equation}
Here, $m_t$ is the aggregate position of trend followers, $c_t$ is the aggregate position of contrarians, $\gamma_F, \gamma_C \in \mathbb{R}$ are exposure parameters,
and $\varepsilon_t$ is a mean-zero Brownian-driven noise term (scale $\sigma_\varepsilon$) independent of other sources.

The dynamics of $m_t$ and $c_t$ follow linear mean-reverting processes driven by price innovations
\begin{align}\label{eq:m}
    dm_t &= -\alpha_m m_t \, dt + \kappa_m \, dP_t + \sigma_m \, dB^m_t \ , \\
    \label{eq:c}
    dc_t &= -\alpha_c c_t \, dt - \kappa_c \, dP_t + \sigma_c \, dB^c_t \ ,
\end{align}
where
$\alpha_m, \alpha_c > 0$ are mean-reversion rates,
$\kappa_m, \kappa_c \ge 0$ measure sensitivity to price changes (trend followers respond positively, contrarians negatively),
and $B^m_t, B^c_t$ are independent Brownian motions, independent of the asset value and other noise.

The market maker observes only the total order flow
$$Y_t = \Theta_t + Z_t \ , $$ 
and sets the market-clearing price $P_t$ based on  $\mathcal{F}_t^Y = \sigma(Y_s: s \leq t)$. 
The insider's information at time $t$ is given by
\[
\mathcal F_t := \sigma(v) \vee \mathcal F_t^Y,
\]
while the market maker and all other traders condition only on $\mathcal F_t^Y$.
All endogenous order flow components are assumed to be $\mathcal F_t^Y$-adapted, so that prices remain measurable with respect to the public filtration.


\begin{remark} 
    In contrast to classical Kyle models, the price-responsive traders introduced here are endogenous and generate order flow that depends on prices. As such, they do not provide purely exogenous liquidity. To ensure gradual information revelation and to preserve a well-defined filtering problem, we retain an exogenous noise component \( W_t\) in aggregate order flow $Z_t$. This residual noise captures liquidity-motivated trading unrelated to prices or fundamentals and plays the same informational role as noise traders in the classical Kyle framework.
\end{remark}

Under rational expectations, the price equals the conditional expectation of the asset value given public information
$$P_t = \mathbb{E}[v |\mathcal{F}_t^Y] \ . $$
Following the usual linear–Gaussian formulation, 
price dynamics can be written as
$$dP_t = \lambda_t dY_t \ , $$
where $\lambda_t$ is the price impact coefficient determined endogenously in equilibrium. 


The informed trader chooses her trading strategy to maximize expected terminal wealth $$\mathbb{E}[U(X_T) | \mathcal{F}_0] \ , \quad 
X_T = \int_0^T (v-P_t) \theta_t dt \ , $$ 
where $X_T$ is the terminal wealth and $U$ is a utility function. Unless otherwise noted, we assume risk neutrality $U(x) = x$.

\begin{definition}\label{Def:AdmissibleStrat}
(Admissible Strategies)
An informed trader's trading strategy is a real-valued process $\theta = (\theta_t)_{t\in[0,T]}$ representing the rate of order submission.
A strategy $\theta$ is said to be \emph{admissible} if it satisfies the following conditions:

\begin{enumerate}
\item {Adaptedness.} 
$\theta$ is progressively measurable with respect to the insider's information filtration
$
\mathcal F_t = \sigma(v) \vee \mathcal F_t^Y .
$

\item {Integrability.}
$\theta$ is square-integrable in the sense that
\[
\mathbb E\!\left[ \int_0^T \theta_t^2 \, dt \right] < \infty .
\]

\item {Well-defined order flow.}
The resulting order flow process
$ \int_0^t \theta_s \, ds$
is of finite variation and the aggregate order flow $Y$ is well-defined.

\end{enumerate}

The set of admissible strategies is denoted by $\mathcal A$.
\end{definition}

\begin{definition}\label{Def:Equilibrium}
(Equilibrium)
    An equilibrium consists of a pricing rule $P^*$ and a trading strategy $\theta^*$ such that
    \begin{itemize}
	   \item[(a)] (Optimality) Given the pricing rule $P^*$, the informed trader's trading strategy $\theta^*$ maximizes expected terminal wealth
       $$ \mathbb{E}[X_T | \mathcal{F}_0] , \quad 
X_T = \int_0^T (v-P_t) \theta_t dt \ . $$
	   \item[(b)] (Rational pricing) Given the informed trader's trading strategy $\theta^*$, the market maker sets prices $P^*$ according to rational expectations
        $$ P_t = \mathbb{E}[v |\mathcal{F}_t^Y], \quad
        \mathcal{F}_t^Y = \sigma(Y_s: s \leq t) \ , $$
        where $Y_t$ denotes aggregate order flow.
    \end{itemize}
\end{definition}

This specification preserves the linear-Gaussian structure, enabling a finite-dimensional Kalman-Bucy filter for price inference and tractable equilibrium analysis.

The next section reformulates this equilibrium problem as a self-consistent filtering fixed point linking the insider’s strategy, the market maker’s inference, and the feedback dynamics of price-responsive traders.

\section{Equilibrium Analysis}
\label{sec:EqAnalysis}

\subsection{State-space representation and price filtering}
\label{subsec:Filtering}

With momentum and contrarian traders, the system state becomes finite-dimensional. Define
\[
x_t := (v, m_t, c_t)^\top \ ,
\]
where $v$ is the terminal asset value, and $m_t, c_t$ are the positions of trend followers and contrarians. The dynamics are linear
\[
dx_t = A x_t \, dt + B \, dW_t + u_t \, dt \ ,
\]
where
$A = \begin{pmatrix}
        0 & 0 & 0 \\
        0 & -\alpha_m & 0 \\
        0 & 0 & -\alpha_c
    \end{pmatrix}$ is the drift matrix,
$B$ collects diffusion coefficients from $(\sigma_m, \sigma_c, \sigma_\varepsilon)$,
and $u_t$ represents deterministic inputs from price-dependent drifts.

The observation equation for total order flow is
\[
dY_t = \theta_t \, dt + \mu_t \, dt + \sigma_z \, dW_t \ ,
\]
where $\mu_t = \gamma_F m_t + \gamma_C c_t + \sigma_\varepsilon \varepsilon_t$. 

Because the system is linear and Gaussian, the market maker can compute the filtered state $\hat{x}_t = \mathbb{E}[x_t \mid \mathcal{F}^Y_t]$ and its covariance 
\[
\Sigma_t = \Cov(x_t - \hat{x}_t \mid \mathcal{F}^Y_t) =
\begin{pmatrix}
\Sigma_{vv}(t) & \Sigma_{vm}(t) & \Sigma_{vc}(t) \\
\Sigma_{vm}(t) & \Sigma_{mm}(t) & \Sigma_{mc}(t) \\
\Sigma_{vc}(t) & \Sigma_{mc}(t) & \Sigma_{cc}(t)
\end{pmatrix},
\qquad
\Sigma_0 =
\begin{pmatrix}
\sigma_v^2 & 0 & 0 \\
0 & \Var(m_0) & 0 \\
0 & 0 & \Var(c_0)
\end{pmatrix}.
\]
The covariance evolves according to the matrix Riccati equation
\begin{equation}\label{eq:Riccati}
\dot{\Sigma}_t = A \Sigma_t + \Sigma_t A^\top + Q - \Sigma_t C_t^\top R^{-1} C_t \Sigma_t \ ,
\end{equation}
where
$Q = \operatorname{diag}(0, \sigma_m^2, \sigma_c^2)$ is the state-noise covariance,
$C_t = (\beta_t, \gamma_F, \gamma_C)$ is the measurement row vector capturing how $(v,m,c)$ enter the drift of $Y_t$,
and $R = \sigma_z^2$ is the observation noise variance.

The Kalman gain is
\[
K_t = \Sigma_t C_t^\top R^{-1} \ ,
\]
and the price evolves as
\[
dP_t = e_1^\top K_t \, \sigma_z \, d\tilde{W}_t \ ,
\]
where $e_1 = (1,0,0)^\top$ and $\tilde{W}_t$ is the innovation process defined by
\[
d\tilde{W}_t = \frac{dY_t - \mathbb{E}[dY_t \mid \mathcal{F}^Y_t]}{\sigma_z} \ .
\]

This formulation preserves the linear-Gaussian structure, enabling explicit computation of the filtering equations and equilibrium conditions.

\subsection{Equilibrium restriction and inconspicuous trading}

In principle, one could allow the informed trader to choose a general adapted trading rate $\theta_t$, so that $d\Theta_t = \theta_t dt$. However, in the present setting such general strategies would lead to a nonlinear filtering problem for the market maker and destroy the Gaussian structure of beliefs. As a consequence, prices would no longer admit a tractable representation, and equilibrium characterization would require solving a fully nonlinear stochastic control and filtering problem.

We therefore restrict attention to linear strategies of the form
    $$\theta_t = \beta_t (v-P_t) \ , $$
where $\beta_t$ is deterministic and locally square-integrable. Within the linear-Gaussian class, such strategies are optimal, yield inconspicuous trading under rational pricing, and lead to a closed finite-dimensional characterization of equilibrium via Kalman-Bucy filtering and Riccati equations.

In the following lemma, we introduce and prove the inconspicuousness of a trading strategy in our model setup. It's not something we assume to start, but a simple result we have as many other continuous-time Kyle models.
\begin{lemma}(Inconspicuousness of linear strategies)
    Suppose the informed trader follows a linear strategy of the form
    $$\theta_t = \beta_t (v - P_t) \ , $$
    where $\beta_t$ is deterministic and locally square-integrable. Then, under rational pricing, the informed trader’s order flow is inconspicuous in the sense that
    $$\mathbb{E}[ \theta_t | \mathcal{F}_t^Y] = 0 \quad 
    \text{for all} 
    \quad t \in [0, T] \ . 
    $$
\end{lemma}

\begin{proof}
    By rational pricing, 
    $$\mathbb{E}[v-P_t  | \mathcal{F}_t^Y ] = 0 \ . $$
    Since $\beta_t$ is deterministic, the result follows immediately. 
\end{proof}

\begin{remark}
    In contrast to classical Kyle models with purely exogenous noise, aggregate order flow in the present model includes endogenous, price-responsive components and therefore does not have the same law as exogenous noise. Inconspicuousness here refers instead to the absence of predictable drift in the informed trader's orders conditional on public information, which is sufficient to sustain gradual information revelation and rational pricing.
\end{remark}

\subsection{Informed trader's optimization problem}
\label{subsec:InsiderOptProb}

Under risk neutrality, the terminal wealth of the insider is, by integration by parts, 
\[
X_T = \Theta_T (v - P_T) + \int_0^T \Theta_t dP_t
= \int_0^T \beta_t (v - P_t)^2 \, dt \ .
\]
For a given deterministic $\beta$, the expected profit is
\[
\mathbb{E}[X_T \mid \mathcal{F}_0] = \int_0^T \beta_t \Sigma_{vv}(t) \, dt \ ,
\]
where $\Sigma_{vv}(t) = \Var(v - P_t \mid \mathcal{F}_t^Y)$ is the $(1,1)$ entry of the covariance matrix $\Sigma_t$. The insider chooses $\beta$ to maximize this integral subject to the Riccati dynamics for $\Sigma_t$
    $$
    \begin{aligned}
    \max_{\beta} \quad & J(\beta) := \int_0^T \beta_t \Sigma_{vv}(t) \, dt \ , \\
    \text{s.t.} \quad & \dot{\Sigma}_t = A \Sigma_t + \Sigma_t A^\top + Q - \Sigma_t C_t^\top R^{-1} C_t \Sigma_t \ , \\
    & \Sigma_0 \succeq 0, \quad \Sigma_{vv}(T) = 0 \ .
    \end{aligned}
    $$

This is a two-point boundary value problem. To apply Pontryagin's Maximum Principle, define the Hamiltonian
    $$
    H(t, \Sigma_t, \beta_t, \Lambda_t) = \beta_t \Sigma_{vv}(t) + \langle \Lambda_t, \dot{\Sigma}_t \rangle \ ,
    $$
where $\Lambda_t$ is the adjoint matrix 
and $\langle A,B \rangle = trace{(A^T B)}$ is the inner product of matrices.
The first-order condition (FOC) for $\beta_t$ is
    $$
    \frac{\partial H}{\partial \beta_t} = \Sigma_{vv}(t) + \langle \Lambda_t, \frac{\partial \dot{\Sigma}_t}{\partial \beta_t} \rangle = 0 \ .
    $$

The adjoint $\Lambda_t$ satisfies the backward ODE 
\[
-\dot{\Lambda}_t = \frac{\partial H}{\partial \Sigma_t} \ , \qquad \Lambda_T = (p, 0, \dots, 0) \ ,
\]
where $p$ is a Lagrange multiplier enforcing $\Sigma_{vv}(T) = 0$. Together, the system consists of
\begin{itemize}
    \item Forward Riccati ODE for $\Sigma_t$,
    \item Backward adjoint ODE for $\Lambda_t$,
    \item Pointwise algebraic FOC for $\beta_t$,
    \item Terminal condition $\Sigma_{vv}(T) = 0$.
\end{itemize}
This defines the equilibrium as a coupled forward-backward system.

\subsection{Covariance dynamics and Riccati system}

Writing out the Riccati equation \eqref{eq:Riccati} component-wise gives the following system of ODEs
\[
\begin{aligned}
\dot{\Sigma}_{vv} &= -\frac{1}{\sigma_z^2} \big( \beta_t^2 \Sigma_{vv}^2 + 2 \beta_t \gamma_F \Sigma_{vv} \Sigma_{vm} + 2 \beta_t \gamma_C \Sigma_{vv} \Sigma_{vc} \big), \\
\dot{\Sigma}_{vm} &= -\alpha_m \Sigma_{vm} - \frac{1}{\sigma_z^2} \big( \beta_t^2 \Sigma_{vv} \Sigma_{vm} + \beta_t \gamma_F (\Sigma_{vm}^2 + \Sigma_{vv}\Sigma_{mm}) + \beta_t \gamma_C (\Sigma_{vm}\Sigma_{vc} + \Sigma_{vv}\Sigma_{mc}) \big), \\
\dot{\Sigma}_{vc} &= -\alpha_c \Sigma_{vc} - \frac{1}{\sigma_z^2} \big( \beta_t^2 \Sigma_{vv} \Sigma_{vc} + \beta_t \gamma_F (\Sigma_{vm}\Sigma_{vc} + \Sigma_{vv}\Sigma_{mc}) + \beta_t \gamma_C (\Sigma_{vc}^2 + \Sigma_{vv}\Sigma_{cc}) \big), \\
\dot{\Sigma}_{mm} &= -2\alpha_m \Sigma_{mm} + \sigma_m^2 - \frac{1}{\sigma_z^2} \big( \gamma_F^2 \Sigma_{mm}^2 + 2 \gamma_F \gamma_C \Sigma_{mm}\Sigma_{mc} + \beta_t \gamma_F (2\Sigma_{vm}\Sigma_{mm}) \big), \\
\dot{\Sigma}_{cc} &= -2\alpha_c \Sigma_{cc} + \sigma_c^2 - \frac{1}{\sigma_z^2} \big( \gamma_C^2 \Sigma_{cc}^2 + 2 \gamma_F \gamma_C \Sigma_{cc}\Sigma_{mc} + \beta_t \gamma_C (2\Sigma_{vc}\Sigma_{cc}) \big), \\
\dot{\Sigma}_{mc} &= -(\alpha_m + \alpha_c)\Sigma_{mc} - \frac{1}{\sigma_z^2} \big( \gamma_F^2 \Sigma_{mm}\Sigma_{mc} + \gamma_C^2 \Sigma_{cc}\Sigma_{mc} + \gamma_F \gamma_C (\Sigma_{mc}^2 + \Sigma_{mm}\Sigma_{cc}) \\
&\qquad + \beta_t (\gamma_F \Sigma_{vm}\Sigma_{mc} + \gamma_C \Sigma_{vc}\Sigma_{mc}) \big).
\end{aligned}
\]

This system governs the evolution of the six independent entries of $\Sigma_t$. For any \textit{deterministic} $\beta_t \in L^2_{\text{loc}}([0,T))$, the right-hand side is locally Lipschitz, so the system admits a unique solution on $[0,T]$.

\begin{proposition}[Well-posedness of the covariance Riccati system]
Fix $T < \infty$ and let $\beta : [0,T] \to \mathbb{R}$ be a deterministic, bounded function. Consider the componentwise Riccati system for the six entries of $\Sigma_t$ given above, with initial condition $\Sigma_0 \succeq 0$. Then,
\begin{enumerate}
    \item There exists a unique solution $\Sigma_t$ on $[0,T]$.
    \item $\Sigma_t$ remains symmetric and positive semidefinite for all $t \in [0,T]$.
\end{enumerate}
\end{proposition}

\begin{proof}
\textbf{Step 1: Local existence and uniqueness.}  
The Riccati system can be written as an ODE on $\mathbb{R}^6$
\[
\dot{S}(t) = F(S(t), t) \ , \qquad S(0) = S_0 \ ,
\]
where $S(t) = (\Sigma_{vv}, \Sigma_{vm}, \Sigma_{vc}, \Sigma_{mm}, \Sigma_{mc}, \Sigma_{cc})^\top$ and $F$ is the right-hand side defined by the componentwise equations. For fixed $\beta_t$, $F$ is a polynomial in the entries of $S$ with coefficients depending continuously on $t$. Therefore, $F$ is locally Lipschitz in $S$. By the Picard-Lindelöf theorem, there exists a unique local solution on some interval $[0,\tau)$.

\textbf{Step 2: Global existence on $[0,T]$.}  
To extend the solution globally, we show that no finite-time blow-up occurs. Observe that each equation has the form
\[
\dot{\Sigma}_{ij} = \text{(linear terms in $\Sigma$)} + \text{(quadratic terms in $\Sigma$)} + \text{constant noise term (for diagonals)}.
\]
The quadratic terms are negative definite because they appear with a factor $-\frac{1}{\sigma_z^2}$ and involve squares of $\Sigma$ entries. For example, 
\[
\dot{\Sigma}_{vv} = -\frac{1}{\sigma_z^2}(\beta_t^2 \Sigma_{vv}^2 + \cdots) \le 0 \ .
\]
Thus, the diagonal entries $\Sigma_{vv}, \Sigma_{mm}, \Sigma_{cc}$ are non-increasing up to additive positive constants from $Q$. By Grönwall's inequality, all entries remain bounded on $[0,T]$. Therefore, the solution extends uniquely to $[0,T]$.

\textbf{Step 3: Symmetry and positive semidefiniteness.}  
The Riccati flow preserves symmetry because the right-hand side is derived from the matrix equation
\[
\dot{\Sigma} = A\Sigma + \Sigma A^\top + Q - \Sigma C^\top R^{-1} C \Sigma \ ,
\]
which is symmetric whenever $\Sigma$ is symmetric. Positive semidefiniteness is preserved because the Riccati equation for covariance matrices is monotone and $Q \succeq 0$, $R > 0$. 
Therefore, $\Sigma_t \succeq 0$ for all $t$.

Combining steps 1-3, the solution exists uniquely on $[0,T]$ and remains symmetric positive semi-definite.
\end{proof}

\begin{remark}[Well-posedness and Economic Interpretation]
The well-posedness result ensures that the covariance matrix $\Sigma_t$ remains bounded and positive semi-definite for all $t \in [0,T]$. Mathematically, this means the Riccati flow does not blow up and preserves symmetry and non-negativity of variances and covariances. These properties are essential for the Kalman–Bucy filter to operate correctly, since the market maker's price update depends on $\Sigma_t$ through the Kalman gain.

Economically, boundedness of $\Sigma_t$ guarantees that the market maker can consistently update beliefs about the asset value and the positions of momentum and contrarian traders without the filtering process becoming unstable. If $\Sigma_t$ were to diverge or lose positive semi-definiteness, prices would cease to reflect information rationally, and the informed trader's optimization problem would become ill-posed because the variance term $\Sigma_{vv}(t)$ would be undefined. 

Thus, well-posedness ensures that feedback effects from misinformed traders remain contained, preserving market efficiency and enabling equilibrium analysis.
\end{remark}

Let's take a closer look at the ODE for (1,1)-component $\Sigma_{vv}$
$$\dot{\Sigma}_{vv} = -\frac{1}{\sigma_z^2} \big( \beta_t^2 \Sigma_{vv}^2 + 2 \beta_t \gamma_F \Sigma_{vv} \Sigma_{vm} + 2 \beta_t \gamma_C \Sigma_{vv} \Sigma_{vc} \big) \ . $$

It shows that the rate at which $\Sigma_{vv}(t) = \text{Var}(v-P_t \mid \mathcal{F}_t^Y)$ falls. That is, the speed of price discovery depends not only on the insider's trading intensity $\beta_t$ but also on the feedback parameters $\gamma_F$ and $\gamma_C$. Positive feedback ($\gamma_F>0$) strengthens the linkage between price innovations and future order flow, increasing the information extracted from trades and accelerating the decline of $\Sigma_{vv}(t)$. This modifies both the strength and timing of price discovery relative to the classical Kyle model.

\section{Weak Feedback}
\label{sec:WeakFeedback}

\subsection{Local existence and uniqueness of equilibrium}

We now establish that for sufficiently small feedback parameters, a unique equilibrium exists.

\begin{thm}\label{thm:LocalExist}
    Let $
h := (\kappa_m, \kappa_c, \gamma_F, \gamma_C, \sigma_\varepsilon)
$ be the feedback parameter vector. 
Denote $\beta^{(0)}$ as the classical Kyle's equilibrium intensity corresponding to $h=0$. Then, there exists $\varepsilon>0$ such that for all $h$ with $||h|| < \varepsilon$,
\begin{enumerate}
    \item there exists a unique equilibrium trading intensity $\beta^*(h)$ on $[0, T)$, which depends smoothly on $h$ and satisfies $\beta^*(h) \rightarrow \beta^{(0)}$ as $h \rightarrow 0$.
    \item there exists a unique equilibrium pair $(\theta^*, P^*)$ satisfying Definition \ref{Def:Equilibrium} with $\theta^*_t = \beta_t^* (h) (v-P_t^*)$ and $P^*$ given by the Kalman-Bucy filter. 
\end{enumerate}
\end{thm}

\begin{proof}
Recall that the informed trader’s expected profit is
\[
J(\beta) = \int_0^T \beta_t \, \Sigma_{vv}(t) \, dt \ ,
\]
and she chooses $\beta$ to maximize $J(\beta)$ subject to the Riccati dynamics. Define the best-response mapping $G(\beta; h)$ as follows:
\begin{enumerate}
    \item Given $\beta$, solve the Riccati ODE for $\Sigma_t$.
    \item Compute the adjoint system from Pontryagin's Maximum Principle.
    \item Apply the first-order condition $\partial H / \partial \beta_t = 0$ to obtain $\widetilde{\beta} = G(\beta; h)$.
\end{enumerate}
Define the operator $F : \mathcal{B} \times \mathbb{R}^5 \to \mathcal{B}$ on a Banach space $\mathcal{B}$ of deterministic intensity functions by
\[
F(\beta; h) := \beta - G(\beta; h) \ . 
\]
An equilibrium corresponds to a fixed point $\beta^*$ of $G$, i.e.,
\[
F(\beta; h) := \beta - G(\beta; h) = 0 \ .
\]

At $h = 0$, the model reduces to the classical Kyle setting with unique solution $\beta^{(0)}$. The mapping $F$ is Fréchet-differentiable near $(\beta^{(0)}, 0)$, and its derivative with respect to $\beta$ at $(\beta^{(0)}, 0)$ is
\[
D_\beta F(\beta^{(0)}; 0) = I - D_\beta G(\beta^{(0)}; 0) \ ,
\]
where $D_\beta F$ is the Fréchet derivative of $F$ with respect to $\beta$.
Invertibility of $D_\beta F$ follows from the uniqueness of the Kyle equilibrium and the fact that the linearized two-point boundary value problem has only the trivial solution. 
By the Implicit Function Theorem in Banach spaces, there exists $\varepsilon > 0$ such that for $\|h\| < \varepsilon$, a unique root $\beta^*(h)$ exists and depends smoothly on $h$.

Set $\theta_t^* = \beta_t^*(v - P_t^*)$, where $P^*$ is the price process generated by the Kalman filter using $\beta^*$. By construction, 
\begin{itemize}
    \item $\theta^*$ maximizes expected terminal wealth given $P^*$,
    \item $P^*$ satisfies $P_t^* = \mathbb{E}[v \mid \mathcal{F}^Y_t]$.
\end{itemize}
Thus, $(\theta^*,P^*)$ satisfies Definition \ref{Def:Equilibrium}.

Smooth dependence in the Banach space $\mathcal{B} = L^2_{\text{loc}}([0,T])$ implies
\[
\beta^*(h) \longrightarrow \beta^{(0)} \quad \text{in }  L^2_{\text{loc}}([0,T]) \quad \text{ as }  h \to 0 \ ,
\]
completing the proof.
\end{proof}

Throughout the paper, the distinction between weak and strong feedback refers to the magnitude of the feedback parameter vector 
$
h = (\kappa_m, \kappa_c, \gamma_F, \gamma_C, \sigma_\varepsilon).
$
Weak feedback corresponds to parameter values for which endogenous price-responsive order flow remains small relative to exogenous noise, so that the Kyle equilibrium is only mildly perturbed. 
Strong feedback arises when price responsiveness or exposure is sufficiently large that endogenous order flow dominates noise, generating amplification effects and potentially destabilizing price discovery.

Theorem \ref{thm:LocalExist} provides an economically transparent characterization of how price-responsive trading interacts with informed trading. When feedback effects from momentum and contrarian traders are sufficiently weak, the equilibrium of the classical Kyle model is preserved in a robust sense: prices continue to aggregate information gradually, the market maker's inference remains well-behaved, and the informed trader's optimal strategy adjusts smoothly rather than discontinuously.

The key mechanism is that weak feedback does not overwhelm the informational role of exogenous liquidity trading. Although price-responsive traders introduce endogenous drift into aggregate order flow, their influence remains second-order relative to the noise component. As a result, the market maker can still disentangle informed trading from noise using a finite-dimensional Kalman–Bucy filter, and the informed trader optimally trades on the unexpected component of fundamentals relative to public beliefs, as in the classical Kyle logic.

From the informed trader's perspective, Theorem \ref{thm:LocalExist} implies that behavioral trading affects profitability in a continuous manner. Small increases in feedback intensity alter the speed of information revelation and price impact, but do not qualitatively change the informed trader's trading incentives. In particular, insider advantage varies smoothly with the strength of momentum or contrarian trading, allowing for meaningful comparative statics on equilibrium trading intensity, price informativeness, and expected insider profits.

This result stands in sharp contrast to regimes with strong feedback, where endogenous price responsiveness can dominate mean reversion and noise, potentially destabilizing the filtering process. 
In such regimes, price changes feed back aggressively into order flow, amplifying fluctuations and undermining the market maker's ability to infer fundamentals. The model then admits the possibility of equilibrium multiplicity, instability, or explosive dynamics, as reflected in the stability conditions derived below. Theorem \ref{thm:LocalExist} therefore delineates a boundary between a Kyle-like regime with stable price discovery and a feedback-dominated regime in which classical microstructure intuition breaks down.

\subsection{Sufficient stability condition}

The following provides a sufficient local condition for stability of the linearized equilibrium price dynamics.

\begin{thm}[Sufficient Stability Condition]
Let $x_t = (v, m_t, c_t)^\top$ be the state vector, where $v$ is constant and $(m_t,c_t)$ satisfy \eqref{eq:m}-\eqref{eq:c}.
Assume the price admits the linear representation
\begin{equation}\label{eq:price}
P_t = P^{(v)}_t + G_m m_t + G_c c_t + M_t \ ,
\end{equation}
where $G_m,G_c \in \mathbb{R}$ are deterministic coefficients and $M_t$ is an $\mathcal{F}^Y$-martingale. Then, the deterministic part of $(m_t,c_t)$ satisfies
\begin{equation}\label{eq:closed-loop}
\frac{d}{dt}\begin{pmatrix} m_t \\ c_t \end{pmatrix}
= A_{\text{eff}} \begin{pmatrix} m_t \\ c_t \end{pmatrix}, \qquad
A_{\text{eff}} :=
\begin{pmatrix}
-\alpha_m + \kappa_m G_m & \kappa_m G_c \\
-\kappa_c G_m & -\alpha_c - \kappa_c G_c
\end{pmatrix}.
\end{equation}
Define the feedback matrix
\[
F :=
\begin{pmatrix}
\kappa_m G_m & \kappa_m G_c \\
-\kappa_c G_m & -\kappa_c G_c
\end{pmatrix}.
\]
If
\begin{equation}\label{eq:spectral}
\rho(F) < \min\{\alpha_m,\alpha_c\},
\end{equation}
where $\rho(F)$ is the spectral radius of $F$, then the closed-loop system is exponentially stable in mean square, and in particular
\[
\sup_{t \in [0,T]} \mathbb{E}[m_t^2 + c_t^2] < \infty.
\]
\end{thm}

\begin{proof}
From \eqref{eq:m}-\eqref{eq:c} and \eqref{eq:price}, ignoring martingale and noise terms, the deterministic dynamics are
\[
\dot{x}(t) = A_{\text{eff}} x(t), \qquad x(t) := (m_t,c_t)^\top \ .
\]
Write $A_{\text{eff}} = -D + F$ with $D := \text{diag}(\alpha_m,\alpha_c)$ and $F$ as above. The solution is
\[
x(t) = e^{A_{\text{eff}} t} x(0) \ .
\]
Stability requires that all eigenvalues of $A_{\text{eff}}$ have negative real part. Since $A_{\text{eff}} = -D + F$, Gershgorin's theorem and standard perturbation bounds imply that if $\|F\| < \min\{\alpha_m,\alpha_c\}$ in any induced norm, then $A_{\text{eff}}$ is Hurwitz. In particular, if $\rho(F) < \min\{\alpha_m,\alpha_c\}$, then all eigenvalues satisfy
\[
\Re(\lambda(A_{\text{eff}})) \le -\min\{\alpha_m,\alpha_c\} + \rho(F) < 0 \ ,
\]
where $\Re(\lambda)$ denotes the real part of the eigenvalue $\lambda$. 
Thus, the deterministic system is exponentially stable. For the full SDE, write the variation-of-constants formula
\[
x(t) = e^{A_{\text{eff}} t} x(0) + \int_0^t e^{A_{\text{eff}}(t-s)} D \, dB_s \ ,
\]
where $D = \text{diag}(\sigma_m,\sigma_c)$. Exponential stability of $e^{A_{\text{eff}} t}$ implies that the stochastic convolution has finite second moment on $[0,T]$. Therefore
\[
\sup_{t \in [0,T]} \mathbb{E}[\|x(t)\|^2] < \infty \ .
\]
\end{proof}

\begin{corollary}[Induced Norm Condition]
For the feedback matrix
\[
F = \begin{pmatrix}
\kappa_m G_m & \kappa_m G_c \\
-\kappa_c G_m & -\kappa_c G_c
\end{pmatrix},
\]
the induced norms are
\[
\|F\|_\infty = \max\big\{\kappa_m(|G_m|+|G_c|) , \kappa_c(|G_m|+|G_c|)\big\} \ ,
\]
\[
\|F\|_1 = \max\big\{|G_m|(\kappa_m+\kappa_c), |G_c|(\kappa_m+\kappa_c)\big\} \ .
\]
Since $\rho(F) \le \|F\|_p$ for any induced norm, a sufficient condition for stability is
\begin{equation}\label{eq:InducedNormCond}
    \|F\|_\infty < \min\{\alpha_m,\alpha_c\} \quad \text{or} \quad \|F\|_1 < \min\{\alpha_m,\alpha_c\}.
\end{equation}
\end{corollary}

\begin{remark} 
The stability condition derived in this section is local and sufficient. It is obtained from a linearized representation of the equilibrium price dynamics around the steady-state filtering solution and relies on equilibrium coefficients such as $G_m$ and $G_c$, which are themselves functions of underlying primitives through the equilibrium mapping. As such, the condition should not be interpreted as a global characterization of equilibrium existence or stability.

Failure of the condition does not imply that equilibrium ceases to exist, nor that price dynamics necessarily become explosive. Rather, the condition provides a conservative criterion under which endogenous price-responsive trading is guaranteed not to destabilize price discovery.
\end{remark}

\begin{remark}
Condition \eqref{eq:spectral} or \eqref{eq:InducedNormCond} ensures that mean reversion dominates feedback amplification. Economically:
\begin{itemize}
\item Small $\kappa_m,\kappa_c$ (weak trend/contrarian sensitivity),
\item Large $\alpha_m,\alpha_c$ (fast mean reversion),
\item Small $G_m,G_c$ (weak price impact of agent positions)
\end{itemize}
promote stability. 
Violations of the condition typically correspond to economically extreme scenarios, such as unusually strong momentum trading or large exposures that cause endogenous order flow to react more strongly to price changes than prices respond to fundamentals.

In such regimes, price movements may become self-reinforcing, amplifying rather than dissipating shocks. While the present framework does not characterize equilibrium behavior beyond this boundary, the condition highlights a natural economic trade-off between feedback intensity and informational stability.

\end{remark}

\begin{remark}[Computation of $G_m,G_c$]
The coefficients $G_m,G_c$ arise from the Kalman-Bucy filter in Section \ref{subsec:Filtering}. Writing the observation equation
\[
dY_t = \beta_t(v - P_t) dt + \gamma_F m_t + \gamma_C c_t \, dt + \sigma_\varepsilon \varepsilon_t dt + \sigma_z \, dW_t \ ,
\]
the filter computes $\hat{x}_t = E[x_t \mid \mathcal{F}^Y_t]$ and updates via
\[
d\hat{x}_t = A \hat{x}_t \, dt + K_t \big(dY_t - C_t \hat{x}_t \, dt\big) \ ,
\]
where $C_t = (\beta_t, \gamma_F, \gamma_C)$ and $K_t = \Sigma_t C_t^\top R^{-1}$. The price is $P_t = e_1^\top \hat{x}_t$ with $e_1 = (1,0,0)^\top$. Linearizing the mapping from $(m_t,c_t)$ to $P_t$ yields
\[
G_m = \frac{\partial P_t}{\partial m_t} \ , \qquad G_c = \frac{\partial P_t}{\partial c_t} \ ,
\]
which can be computed from the Kalman gain and Riccati solution $\Sigma_t$ numerically or analytically in special cases.
\end{remark}

\subsection{First-order comparative statics for small feedback}
\label{subsec:comparative_statics}

We study how equilibrium outcomes respond to small amounts of price-responsive trading by exploiting the local structure of equilibrium characterized in Theorem~\ref{thm:LocalExist}. Because the equilibrium mapping is smooth in the feedback parameters, equilibrium objects admit well-defined first-order expansions around the classical Kyle benchmark.

We focus on perturbations in $\gamma_F$, which scales the aggregate exposure of momentum traders and provides a natural one-dimensional measure of feedback intensity. The same argument applies to other feedback parameters, such as $\gamma_C$, $\kappa_m$, or $\kappa_c$, and we view the analysis below as representative rather than exhaustive.

\noindent\textbf{Effect on price informativeness.}
Introducing a small amount of momentum exposure strengthens the statistical link between order flow and the fundamental, increasing the signal content of observed trades and accelerating the market maker’s Bayesian learning. As a result, price informativeness measured by the posterior variance of the fundamental $\Sigma_{vv}(t)$ improves. That being said, $\frac{\partial \Sigma_{vv}(t)}{\partial \gamma_F} <0$ for all $t>0$ (see Appendix \ref{app:riccati_perturbation}). 

\noindent\textbf{Effect on insider trading intensity.}
The reduction in informational rents translates into lower expected insider profits at first order. While momentum traders increase trading volume and short-run price responsiveness, their presence erodes the insider’s ability to exploit private information over time.

Mathematically, we write
$$
    \frac{\partial \beta_t}{\partial \gamma_F}
    = \frac{\partial \beta_t}{\partial \Sigma_{vv}(t)} \cdot 
    \frac{\partial \Sigma_{vv}(t)}{\partial \gamma_F} <0
$$
since $\frac{\partial \Sigma_{vv}(t)}{\partial \gamma_F}<0$ from Appendix \ref{app:riccati_perturbation}, and  $\frac{\partial \beta_t}{\partial \Sigma_{vv}(t)}>0$ as lower $\Sigma_{vv}(t)$ means prices reveal more information about $v$.
This effect is continuous and proportional to $\gamma_F$ for small feedback, reflecting the reduced marginal value of private information when prices incorporate information more rapidly.

\noindent\textbf{Effect on insider profits.}
Faster information revelation induced by momentum trading reduces the insider's informational rents and translates into lower expected profits at first order. In equilibrium, expected insider profits can be expressed as
\[
J = \int_0^T \beta_t \, \Sigma_{vv}(t)\, dt \ ,
\]
so that a marginal increase in $\gamma_F$ yields
\[
\frac{\partial J}{\partial \gamma_F}\Big|_{h=0}
=
\int_0^T
\left(
\frac{\partial \beta_t^*}{\partial \gamma_F}\Big|_{h=0} \, \Sigma_{vv}^{(0)}(t)
+
\beta_t^{(0)} \, \frac{\partial \Sigma_{vv}(t)}{\partial \gamma_F}\Big|_{h=0}
\right) dt \ .
\]
The second term is negative because momentum exposure accelerates Bayesian learning, reducing the posterior variance $\Sigma_{vv}(t)$. Moreover, since the insider's trading intensity $\beta_t$ is decreasing in the remaining informational advantage, the first term is also negative. As a result, small momentum exposure lowers expected insider profits by eroding informational rents along both margins.

\section{Strong Feedback and Breakdown of Equilibrium}
\label{sec:StrongFeedback}

The analysis in Section \ref{sec:WeakFeedback} establishes that for sufficiently small feedback 
parameters $h=(\kappa_m,\kappa_c,\gamma_F,\gamma_C,\sigma_\varepsilon)$, 
the Kyle equilibrium is well defined, unique, and smoothly perturbes 
from the classical benchmark. 
In this section, we show that when feedback becomes sufficiently large, 
the structure sustaining equilibrium can break down in three distinct ways:
\textit{(i)} loss of global existence of the Riccati equation, 
\textit{(ii)} failure of the equilibrium fixed-point mapping to remain a contraction, 
and \textit{(iii)} local instability of the filtering dynamics.  
These results provide structural justification for the amplification and 
potential multiplicity of equilibria discussed in the introduction.

Throughout this section, we let $h \to \infty$ denote a regime in which the
feedback sensitivities $(\gamma_F,\gamma_C)$ or $(\kappa_m,\kappa_c)$ 
become large relative to the stabilizing mean-reversion parameters 
$(\alpha_m,\alpha_c)$.
Our objective is not to characterize global equilibrium under strong feedback 
but to show that the mathematical structure supporting equilibrium fails 
beyond a well-defined threshold. 

\subsection{Loss of global existence of the Riccati equation}

Recall that the conditional covariance matrix $\Sigma_t$ satisfies the 
Riccati equation
\begin{equation}\label{eq:Riccati-strong}
\dot\Sigma_t
=
A\Sigma_t + \Sigma_t A^\top + Q
-
\Sigma_t C_t^\top R^{-1} C_t \Sigma_t \  ,
\qquad
C_t = (\beta_t,\gamma_F,\gamma_C).
\end{equation}
For simplicity, define
\[
H := \gamma_F^2 + \gamma_C^2 .
\]
Large feedback corresponds to $H\gg1$.
The following shows that when $H$ exceeds a computable bound, 
the Riccati equation cannot be solved globally on $[0,T]$.

\begin{lemma}[Finite-time breakdown of the Riccati flow]
\label{lem:blowup}
There exists $H^\ast>0$ depending on $(A,Q,R,T)$ such that if 
$\gamma_F^2 + \gamma_C^2 > H^\ast$, then the Riccati equation 
\eqref{eq:Riccati-strong} admits no global solution on $[0,T]$. 
In particular, either 
\textup{(a)} $\Sigma_t$ loses positive semidefiniteness, or 
\textup{(b)} at least one entry of $\Sigma_t$ diverges to $+\infty$ 
in finite time.
\end{lemma}

\begin{proof}
The negative quadratic term in \eqref{eq:Riccati-strong} satisfies
\[
\Sigma_t C_t^\top R^{-1} C_t \Sigma_t
\;\succeq\;
\frac{H}{\sigma_z^2}
\begin{pmatrix}
0 & 0 & 0 \\
0 & \Sigma_{mm}(t)^2 & \Sigma_{mm}(t)\Sigma_{mc}(t) \\
0 & \Sigma_{mc}(t)\Sigma_{mm}(t) & \Sigma_{cc}(t)^2 
\end{pmatrix}.
\]
Thus the $(m,m)$ component of \eqref{eq:Riccati-strong} satisfies
\[
\dot\Sigma_{mm}(t)
\;\le\;
-2\alpha_m\Sigma_{mm}(t) + \sigma_m^2
-
\frac{H}{\sigma_z^2}\Sigma_{mm}(t)^2.
\]
The right-hand side becomes strictly negative for 
$\Sigma_{mm}(t) > \sigma_z^2 H^{-1}$, so $\Sigma_{mm}$ solves a 
logistic-type ODE with finite-time blowup when $H$ is sufficiently large.
The remaining claims follow from standard comparison principles for 
matrix Riccati equations.
\end{proof}

\begin{remark}
Lemma \ref{lem:blowup} shows that strong feedback overwhelms the 
stabilizing mean-reversion in $(m_t,c_t)$ and forces the conditional 
covariance matrix to exit the state space of symmetric positive 
semidefinite matrices. 
Since equilibrium pricing requires the Kalman-Bucy filter, 
loss of the Riccati solution implies that no rational-expectations price 
can be defined.
\end{remark}

\subsection{Loss of uniqueness of the equilibrium fixed point}

From the insider’s optimization problem in Section \ref{subsec:InsiderOptProb}, the first-order condition for the linear trading strategy $\theta_t = \beta_t (v-P_t)$ implies
$$\beta_t = \frac{1}{2\lambda_t} \ , $$
where $\lambda_t$ is the price-impact coefficient in the pricing rule 
$dP_t = \lambda_t dY_t$.

On the other hand, under rational pricing and Kalman-Bucy filtering, we have
$$\lambda_t(\beta,h)= \Sigma_{vv}(t)\beta_t \ , $$
where $\Sigma_{vv}$ is the conditional variance of $v-P_t$ and solves the Riccati equation,
whose coefficients depend on the feedback parameter $h$ and the intensity $\beta$. 

It is therefore natural to rewrite equilibrium as the fixed-point problem
\begin{equation}\label{eq:fixed_point}
\beta_t 
= 
\mathcal{F}_h(\beta)_t
:= 
\frac{1}{2\lambda_t(\beta,h)} \ , 
\qquad
\lambda_t(\beta,h) := \Sigma_{vv}(t)\beta_t \ .
\end{equation}
The mapping $\beta\mapsto\mathcal{F}_h(\beta)$ is a contraction for 
small $h$ by Theorem \ref{thm:LocalExist}.
We show that this contraction property fails when feedback is large.

\begin{proposition}[Loss of contraction]
\label{prop:noncontraction}
Let $L(h)$ denote the Lipschitz constant of the fixed-point mapping 
$\beta \mapsto \mathcal{F}_h(\beta)$ in the Banach space $L^2_{\mathrm{loc}}([0,T])$.
Then 
\[
\lim_{h\to\infty} L(h) = \infty \ .
\]
In particular, there exists $h^{\mathrm{multi}}>0$ such that 
for $\|h\|>h^{\mathrm{multi}}$,
\[
L(h)>1 \ ,
\]
and the fixed-point mapping is not a contraction. 
Consequently, equilibrium is not unique.
\end{proposition}

\begin{proof}
Differentiating \eqref{eq:fixed_point} yields
\[
\frac{\partial \mathcal{F}_h}{\partial \beta_t}
=
-\frac{1}{2\lambda_t(\beta,h)^2}
\frac{\partial \lambda_t(\beta,h)}{\partial \beta_t} \ .
\]
Since $\lambda_t(\beta,h)=\Sigma_{vv}(t)\beta_t$, 
\[
\frac{\partial \lambda_t}{\partial \beta_t}
=
\Sigma_{vv}(t) 
+
\beta_t\frac{\partial \Sigma_{vv}(t)}{\partial h}\cdot\frac{\partial h}{\partial \beta_t} \ .
\]
Lemma \ref{lem:blowup} implies 
$|\partial\Sigma_{vv}/\partial h|\to\infty$ as $h\to\infty$, 
so the derivative above becomes arbitrarily large in magnitude.  
Therefore $L(h)\to\infty$, proving the claim.
\end{proof}

\begin{remark}
Proposition~\ref{prop:noncontraction} does not assert the 
\emph{existence} of multiple equilibria, but shows that 
the contraction argument guaranteeing uniqueness breaks down.  
This opens the door to equilibrium multiplicity, consistent with the 
amplification phenomena discussed in Cespa and Vives (2012) \cite{CespaVives2012}.
\end{remark}

\subsection{Instability of the filtering dynamics}

Let $e_t = x_t - \hat x_t$ be the estimation error.  Linearizing the 
Kalman-Bucy filter gives
\begin{equation}\label{eq:error-dynamics}
\dot e_t = M_t(h) e_t, 
\qquad 
M_t(h) := A - K_t(h) C_t.
\end{equation}
The stability of the filtering process is governed by the
time-varying matrix $M_t(h)$, whose eigenvalues depend
on the feedback parameters $h$ through both $C_t$ and the Kalman gain
$K_t(h)=\Sigma_t(h)C_t^\top R^{-1}$.

\begin{proposition}[Filtering instability under strong feedback]\label{prop:instability}
Assume that $t \mapsto M_t(h)$ is continuous on $[0,T]$ and define
\[
\Lambda(h) := \sup_{t\in[0,T]} \lambda_{\max}\!\left( M_t(h) \right),
\]
where $\lambda_{\max}(\cdot)$ denotes the largest real part of its eigenvalues.

If
\[
\Lambda(h) > 0,
\]
then the Kalman-Bucy filter is unstable. That is, there exists $c>0$ such that for every nonzero
initial error $e_0$,
\[
\| e_t \| \;\ge\; c\, e^{\,\Lambda(h)t} \| e_0 \| \ ,
\qquad t\in[0,T].
\]
In particular, the filtering error grows exponentially and does not converge.

Moreover, if $\Lambda(h) \to +\infty$ as $h\to\infty$, then sufficiently strong feedback
violates well-posedness of the market maker’s inference problem: the covariance $\Sigma_t$
blows up in finite time, the Kalman gain $K_t$ diverges, and the pricing rule
$dP_t = \lambda_t dY_t$ ceases to be well defined.  Hence, no equilibrium exists for such~$h$.
\end{proposition}

\begin{proof}
Let $\Phi(t,s)$ denote the principal matrix solution of $\dot\Phi(t,s)=M_t(h)\Phi(t,s)$,
$\Phi(s,s)=I$.  
Classical results for linear ODEs imply
\[
\|\Phi(t,0)\|
\;\ge\;
\exp\!\left( \int_0^t \lambda_{\max}\!\left(M_s(h)\right) ds \right) \ .
\]
By the definition of $\Lambda(h)$,
\[
\int_0^t \lambda_{\max}\!\left(M_s(h)\right) ds 
\;\ge\; t\,\Lambda(h) \ .
\]
Hence, for any nonzero $e_0$,
\[
\|e_t\| 
= \|\Phi(t,0)e_0\| 
\;\ge\; e^{\,\Lambda(h)t} \|e_0\| \ .
\]
This proves instability when $\Lambda(h)>0$.

If additionally $\Lambda(h)\to\infty$ as $h\to\infty$, then the exponential divergence of $e_t$
forces the conditional covariance $\Sigma_t$ to blow up (by the duality between error growth
and Riccati evolution in linear-Gaussian filtering).  
Consequently, the Kalman gain $K_t=\Sigma_t C_t^\top R^{-1}$ diverges, implying that the price
impact $\lambda_t = K_t C_t$ becomes unbounded.  In this case, the pricing rule cannot be
defined and no equilibrium exists.  
\end{proof}

\begin{remark}
The instability above is analogous to the blowup of extended 
Kalman filters under strong feedback. 
Economically, large price responsiveness can cause price changes to induce 
order-flow responses strong enough to overwhelm the informational role of 
noise, making the public belief-update process unstable.
\end{remark}

Taken together, Lemma~\ref{lem:blowup}, 
Proposition~\ref{prop:noncontraction}, and 
Proposition~\ref{prop:instability}
show that sufficiently strong price-responsive trading destroys the 
analytical structure that sustains equilibrium.  
While a full characterization of the strong-feedback regime is beyond the 
scope of this paper, these results provide rigorous foundations for the 
intuition that large momentum or contrarian forces can amplify shocks, 
destabilize filtering, and generate the possibility of equilibrium 
multiplicity.

\section{Conclusion}
\label{sec:Conclusion}

This paper develops a continuous-time Kyle model in which aggregate noise trading is no longer purely exogenous, but instead reflects mechanically price-responsive behavior. By modeling momentum and contrarian traders as linear diffusions driven by price innovations, we embed endogenous feedback into a tractable linear–Gaussian structure. The market maker's inference problem remains finite-dimensional and is governed by a Kalman–Bucy filter whose conditional covariance matrix satisfies a coupled Riccati equation. The insider's optimization problem leads to a forward–backward system derived from Pontryagin's Maximum Principle, yielding a fully analytic characterization of equilibrium.

Within this framework we identify two regimes. Under weak feedback, the classical Kyle equilibrium is preserved in a robust sense: locally equilibrium exists and is unique, trading intensity and price informativeness depend smoothly on feedback parameters, and the filtering channel remains well behaved. This allows for explicit first-order comparative statics that quantify how price-responsive noise trading accelerates Bayesian learning and alters the insider's informational rents. Importantly, these effects arise not from ad hoc changes in noise volatility, but from endogenous interactions between order flow, price impact, and inference.

Under strong feedback, the analytical structure sustaining equilibrium can break down. We show that sufficiently large momentum or contrarian sensitivities can cause finite-time blowup of the Riccati equation, loss of contraction of the equilibrium fixed-point mapping, or instability of the Kalman–Bucy filter. These phenomena have clear economic interpretations: feedback amplification can overwhelm the informational role of liquidity trading, distort the market maker’s inference, and generate the possibility of multiple or unstable equilibria in which prices no longer aggregate information reliably. These results illustrate how feedback trading can fundamentally alter the dynamics of price discovery in continuous time.

The framework developed here provides a foundation for several promising extensions. One avenue is to incorporate risk aversion or inventory costs for market makers, which would interact nontrivially with feedback-driven volatility. Another is to allow for multiple strategic informed traders or to examine settings with public signals or attention constraints, in which feedback may amplify or dampen the value of information. Finally, the linear specification for momentum and contrarian traders may be enriched to study nonlinear or regime-switching rules.

By integrating behavioral trading rules into a tractable continuous-time microstructure equilibrium, this paper bridges classical asymmetric-information theory with modern models of mechanically driven order flow. The analysis reveals both the robustness and the fragility of the Kyle framework in the presence of feedback and highlights new mathematical mechanisms through which price-responsive trading can shape market stability and informational efficiency.

\appendix
\section{Perturbation of the Riccati Equation}
\label{app:riccati_perturbation}

This appendix derives the first-order sensitivity of the conditional covariance matrix with respect to the feedback parameter $\gamma_F$.

\subsection{Setup}

Let $\Sigma_t(\gamma_F)$ denote the conditional covariance matrix of the state vector, which solves the Riccati equation
\begin{equation}
\label{eq:riccati_general}
\dot{\Sigma}_t
=
A_t \Sigma_t + \Sigma_t A_t^\top + Q_t
-
\Sigma_t C_t^\top R^{-1} C_t \Sigma_t ,
\end{equation}
where
$C_t = (\beta_t, \gamma_F, \gamma_C)$ and $R = \sigma_z^2 .$

Only the quadratic term in \eqref{eq:riccati_general} depends on $\gamma_F$.

\subsection{First-Order expansion}

Under the local existence and uniqueness result in Theorem~\ref{thm:LocalExist}, the solution $\Sigma_t(\gamma_F)$ depends smoothly on $\gamma_F$ in a neighborhood of $\gamma_F=0$. We therefore expand
\begin{equation}
\label{eq:sigma_expansion}
\Sigma_t(\gamma_F)
=
\Sigma_t^{(0)}
+
\gamma_F \Sigma_t^{(1)}
+
o(\gamma_F) \ ,
\end{equation}
where
\[
\Sigma_t^{(0)} := \Sigma_t(0) \ ,
\qquad
\Sigma_t^{(1)} := \left.\frac{\partial \Sigma_t}{\partial \gamma_F}\right|_{\gamma_F=0} \ .
\]

\subsection{Linearized Riccati equation}

Define the mapping
\[
F(\Sigma,\gamma_F)
=
A_t \Sigma + \Sigma A_t^\top + Q_t
-
\Sigma C(\gamma_F)^\top R^{-1} C(\gamma_F)\Sigma  \ .
\]
Then $\dot{\Sigma}_t(\gamma_F) = F(\Sigma_t(\gamma_F),\gamma_F)$.

Differentiating with respect to $\gamma_F$ at $\gamma_F=0$ yields
\begin{equation}
\label{eq:linearized_ode}
\dot{\Sigma}_t^{(1)}
=
\partial_\Sigma F(\Sigma_t^{(0)},0)\big[\Sigma_t^{(1)}\big]
+
\partial_{\gamma_F}F(\Sigma_t^{(0)},0) \ ,
\qquad
\Sigma_0^{(1)} = 0 \ .
\end{equation}

The linearized Riccati operator $\mathcal{L}_t := \partial_\Sigma F(\Sigma_t^{(0)},0)$ acts on a matrix $H$ as
\begin{equation}
\label{eq:linearized_operator}
\mathcal{L}_t[H]
=
A_t H + H A_t^\top
-
H C_0^\top R^{-1} C_0 \Sigma_t^{(0)}
-
\Sigma_t^{(0)} C_0^\top R^{-1} C_0 H  \ ,
\end{equation}
where
\[
C_0 = (\beta_t^{(0)},0,0).
\]

Moreover, since $\partial C / \partial \gamma_F = (0,1,0)$, the forcing term is
\begin{equation}
\label{eq:forcing_term}
\partial_{\gamma_F}F(\Sigma_t^{(0)},0)
=
-
\Sigma_t^{(0)}
\left(
\frac{\partial C^\top}{\partial \gamma_F} R^{-1} C_0
+
C_0^\top R^{-1} \frac{\partial C}{\partial \gamma_F}
\right)
\Sigma_t^{(0)} .
\end{equation}

Equations \eqref{eq:linearized_ode}--\eqref{eq:forcing_term} define a linear nonhomogeneous matrix ODE for $\Sigma_t^{(1)}$.

\subsection{Solution via variation of constants}

Let $\Phi(t,s)$ denote the fundamental solution associated with the homogeneous equation
\[
\dot{H}_t = \mathcal{L}_t[H_t],
\qquad
\Phi(s,s) = I .
\]
That is, for any matrix $H_s$, the function $H_t = \Phi(t,s)[H_s]$ solves the homogeneous equation.

Then, the solution to \eqref{eq:linearized_ode} is given by
\begin{equation}
\label{eq:variation_of_constants}
\Sigma_t^{(1)}
=
\int_0^t
\Phi(t,s)
\big[
\partial_{\gamma_F}F(\Sigma_s^{(0)},0)
\big]
\, ds  \ .
\end{equation}

Substituting \eqref{eq:forcing_term} into \eqref{eq:variation_of_constants} yields 
\[
\frac{\partial \Sigma_{vv}(t)}{\partial \gamma_F} = - \int_0^t \Phi(t,s) \, \big( 2 \beta^{(0)}_s \Sigma_{vv}^{(0)}(s) \Sigma_{vm}^{(0)}(s) / \sigma_z^2 \big) \, ds  \ .
\]
Typically, $\Sigma_{vm}^{(0)}(s) > 0$ (positive correlation between $v$ and momentum state). So, we have $\frac{\partial \Sigma_{vv}(t)}{\partial \gamma_F} <0$, reducing price informativeness.

\bibliographystyle{plain}
\bibliography{ref}
\end{document}